\documentclass[letterpaper,11pt]{article}
\usepackage[affil-it]{authblk}
\usepackage[usenames,dvipsnames]{xcolor}
\usepackage{amsfonts}
\usepackage{amsmath,amsthm,amssymb,dsfont,pifont}
\usepackage{enumerate}
\usepackage[english]{babel}
\usepackage{graphicx}	
\usepackage{subcaption}
\usepackage[margin=1in]{geometry}
\usepackage{url}
\usepackage{todonotes}
\usepackage{bbm}
\usepackage{caption}
\usepackage{comment}
\usepackage{cite}
\usepackage{physics}

\usepackage{thm-restate}
\usepackage{epsfig}
\usetikzlibrary{shapes.symbols,patterns} 
\usepackage{hyperref}[breaklinks]
\hypersetup{colorlinks=true,citecolor=blue,linkcolor=blue,filecolor=blue,urlcolor=blue}
\usepackage{cleveref}

\usepackage{nicefrac}
\usepackage{mathtools}

\usepackage{algorithm}
\usepackage{algorithmic}
 
\theoremstyle{plain}
\newtheorem{theorem}{Theorem}
\newtheorem{lemma}[theorem]{Lemma}

\newtheorem{fact}[theorem]{Fact}
\newtheorem{corollary}[theorem]{Corollary}
\newtheorem{proposition}[theorem]{Proposition}

\theoremstyle{definition}
\newtheorem{definition}[theorem]{Definition}

\newcommand*{\1}{\mathbbm{1}}
\newcommand*{\PP}{\mathop{\mathbb{P}}}
\newcommand*{\E}{\mathop{\mathbb{E}}}

\newcommand*{\R}{\mathbb{R}}

\newcommand{\F}{\mathbb{F}}

\newcommand*{\eps}{\varepsilon}

\usepackage{preamble}
\def\Fnn{\F_2^{2n}}

\newcommand{\defeq}{\overset{\textrm{\tiny def}}{=}}
\usepackage[normalem]{ulem}

\usepackage{enumitem}
\usepackage{authblk}
\usepackage{thm-restate}

\title{Tolerant testing of stabilizer states with a polynomial gap via a generalized uncertainty relation}

\author[1]{Zongbo (Bob) Bao}
\author[2]{Philippe van Dordrecht}
\author[1]{Jonas Helsen}
\affil[1]{QuSoft and CWI, Amsterdam, The Netherlands}
\affil[2]{QuSoft and University of Amsterdam, Amsterdam, The Netherlands}
\date{Feb 21, 2025}
\allowdisplaybreaks

\begin{document}
\maketitle

\begin{abstract}
We prove a conjecture of Arunachalam \& Dutt \cite{arunachalam2024tolerant} on the existence of a tolerant stabilizer testing algorithm, and achieve an \emph{exponential} improvement in the parameters of the tester. Key to our argument is a generalized uncertainty relation for sets of Pauli operators, based on the Lov{\'a}sz theta function. 
\end{abstract}
\section{Introduction}
We consider the problem of stabilizer testing: given copies of a state $\ket{\psi}$ with the promise that it is either $\epsilon_1$-close to a stabilizer state (as measured by state fidelity) or $\epsilon_2$-far away, decide which one is the case. This problem has garnered significant interest in recent years following the introduction of the Bell difference sampling algorithm in \cite{gross2021schur}. There it was proven that Bell difference sampling gives rise to a one-sided tester (i.e. setting $\epsilon_1=1$ and $\epsilon_2<1$). Their analysis was improved by \cite{grewal2024improved}, who proved that Bell difference sampling gives rise to a tester provided that  $\epsilon_2 \leq C\epsilon_1^6$. However, they also required that $\epsilon_1$ was larger than some critical threshold. This latter requirement was relaxed by \cite{arunachalam2024tolerant}, allowing for arbitrary $\epsilon_1$ (this regime is called tolerant testing). However, their argument only obtained a weaker gap relation of the form $\epsilon_2 \leq \exp{-O\p{1/\epsilon_1^c}}$
, and furthermore it turned out their argument relied on an unproven statement in additive combinatorics, making their ultimate result conjectural (see Conjecture 1.2 in \cite{arunachalam2024tolerant}\footnote{More specifically the second version currently on the arXiv. The first version (lemma 4.17) contains contains a misstatement of a result from \cite{viola2011selected}, leading to the requirement of Conjecture 1.2 in the second version.}).\\

\noindent In this paper we both sidestep their need for a conjecture, and improve their gap relation to a \emph{polynomial} one. This is the polynomial gap mentioned in the title. Concretely we prove the following result:

\begin{restatable}{theorem}{theoremone}\label{theoremone}
Let $\epsilon_1,\epsilon_2 \in [0,1]$ and let $\ket{\psi}$ be a state with stabilizer fidelity 
\begin{equation}
\mathcal{F}_S(\ket{\psi}) := \max_{\ket{S}\in \mathrm{STAB_n}} \abs*{\braket{\psi}{S}}^2,
\end{equation}
either (1) larger than $\epsilon_1$ or (2) smaller than $\epsilon_2$. If $\epsilon_2\leq C'\epsilon_1^{672}$ for some constant $C'$, then $O\big(\epsilon_1^{-12}\big)$ rounds of Bell difference sampling can distinguish between case (1) and case (2) with probability greater than $2/3$.
\end{restatable}

\noindent The polynomial degree we achieve is likely highly suboptimal. We aim to improve it in future work. To achieve this result we build on the work of \cite{arunachalam2024tolerant}, combined with a generalized uncertainty relation involving the Lov{\'a}sz theta function, which we derive from a result on Hamiltonian optimization due to Hastings \& O'Donnell \cite{hastings2022optimizing}, which we think may be of independent interest. \\

\noindent \textbf{Remark:} After we showed the authors of \cite{arunachalam2024tolerant} a summary of our result, they derived their own version of our result based on an earlier note by Sergey Bravyi.
Their argument, which uses different techniques, but proceeds essentially along the same lines, is given in \cite{NOTE_SA}. A third, independent, proof of our result was also posted \cite{2024merhdad} shortly after we posted a first version of this paper on the arXiv. This argument proceeds along different lines from ours, achieving a slightly worse polynomial degree. 
\section{Preliminaries}
In this section we quote some background facts on graph theory (in particular on the Lov{\'a}sz theta) and Bell difference sampling. 

\subsection{Graph theoretic notions}
\noindent\textbf{Notation:} Given a (simple) graph $\Gamma=\p{V\p{\Gamma},E\p{\Gamma}}$ we denote its vertex 
  set by $V(\Gamma)$ and its edge set by $E(\Gamma)$. 
We use $K_n$ to denote the complete graph on $n$ vertices and
  $\overline{\Gamma}$ to represent the complement graph of 
  $\Gamma$, which shares the same vertex set of $\Gamma$,
  but has complementary edges, i.e. $E\p{\overline{\Gamma}}\cap E\p{\Gamma} =\varnothing$
  and $E\p{\overline{\Gamma}}\cup E\p{\Gamma} = 
  E\p{K_{\abs*{V(\Gamma)}}}$.
For graphs $\Gamma_1$ and $\Gamma_2$, say $\Gamma_1 \subseteq \Gamma_2$ if and only if $V\p{\Gamma_1} = V\p{\Gamma_2}$ and
  $E\p{\Gamma_1} \subseteq E\p{\Gamma_2}$.\\

\noindent\textbf{Disjoint union:} Graphs $\Gamma_1=\p{V_1,E_1}$ and $\Gamma_2=\p{V_2,E_2}$,
  have a disjoint union, denoted by $\Gamma_1 \sqcup \Gamma_2$.
  It is defined as the graph with $V\p{\Gamma_1\sqcup\Gamma_2} = V_1
  \cup V_2$, where we assume $V_1$ and $V_2$ are disjoint and $E\p{\Gamma_1\sqcup\Gamma_2} = E\p{\Gamma_1}\cup E\p{\Gamma_2}$.\\
  
\noindent\textbf{Strong product:} Graphs $\Gamma_1=\p{V_1,E_1}$ and $\Gamma_2=\p{V_2,E_2}$,
  have a \emph{strong product}, denoted by $\Gamma_1 \boxtimes \Gamma_2$.
  It is defined as a graph with vertex set $V\p{\Gamma_1 \boxtimes \Gamma_2} = V_1
  \times V_2$. Distinct vertices 
  $\p{u_1,v_1},\p{u_2,v_2}\in V_1\times V_2$ are connected by an edge
  if and only if:
\begin{enumerate} 
    \item $u_1=u_2$ and $v_1$ is connected with $v_2$ in $\Gamma_2$, or
    \item $v_1=v_2$ and $u_1$ is connected with $u_2$ in $\Gamma_1$, or
    \item $u_1$ is connected with $u_2$ in $\Gamma_1$ and 
      $v_1$ is connected with $v_2$ in $\Gamma_2$.\\
\end{enumerate}

\noindent\textbf{Lov{\'a}sz theta function:} Lov{\'a}sz \cite{lovasz1979shannon} introduced the following graph parameter in the context of classical Shannon theory:
\begin{definition}[Lov{\'a}sz theta function]
  For any graph $\Gamma=\p{[n],E}$, the Lov{\'a}sz theta is
    the number
  \begin{equation}
    \vartheta\p{\Gamma} = \max_{\rho \in \R^{n\times n}}
      \lrs{
        \Tr\p{\rho \mathbb{J}}
        \ \Big\vert \ 
        \rho \ge 0,
        \Tr \p{\rho} = 1,
        \rho_{jk}=0\  \forall {j,k}\in E
      },
  \end{equation}
  where $\mathbb{J}$ denotes the all-$1$'s matrix.
\end{definition}

We will not need anything particularly sophisticated regarding the Lov{\'a}sz theta function in our proofs, merely the following facts:

\begin{fact}[Theorem 7 of \cite{lovasz1979shannon}]
  \label{fact:theta-product}
  For graphs $\Gamma_1$ and $\Gamma_2$, 
    $\vartheta\p{\Gamma_1\boxtimes \Gamma_2} = 
      \vartheta\p{\Gamma_1} \cdot \vartheta\p{\Gamma_2}$.
\end{fact}

\begin{fact}[Section 18 of \cite{knuth1993sandwich}]
  \label{fact:theta-sum}
  For graphs $\Gamma_1$ and $\Gamma_2$, 
    $\vartheta\p{\Gamma_1 \sqcup \Gamma_2} = 
      \vartheta\p{\Gamma_1} + \vartheta\p{\Gamma_2}$.
\end{fact}

\begin{fact}[Monotonicity of Lov{\'a}sz theta function]
  \label{fact:theta-mono}
  For graphs $\Gamma_1$ and $\Gamma_2$ such that 
    $\Gamma_1\subseteq \Gamma_2$, we have that
    $\vartheta\p{\Gamma_1} \ge \vartheta\p{\Gamma_2}$.
\end{fact}
The last fact follows directly from the definition, as more edges will result in more constraints in the maximization.

\subsection{Weyl operators and distributions}
We recall some standard concepts related to stabilizer states and stabilizer testing. We will consistently refer to Hermitian (unsigned) $n$-qubit Pauli operators as Weyl operators:
\begin{definition}[Weyl operator]
For $x=(x_1,x_2)\in \F_2^n\times \F_2^n =\F_2^{2n} $, the Weyl operator $W_x$ is defined as
\begin{equation}
W_x=i^{x_1\cdot x_2}\bigotimes_{i=1}^nX^{x_{1,i}}Z^{x_{2,i}}.
\end{equation}
\end{definition}
We will occasionally intentionally identify (i.e. confuse) binary vector spaces and sets of Weyl operators. \\

If we want to know whether two Weyl operators commute, we calculate their symplectic inner product (which is a genuine symplectic form on $\mathbb{F}_2^{2n}$):
\begin{definition}The \textit{symplectic inner product} between two vectors $x,y\in \F_2^{2n}$ is the bilinear form
\begin{equation}
  [x,y]=\langle x_1,y_2\rangle+\langle x_2,y_1\rangle,
\end{equation}
where $x=(x_1,x_2), y=(y_1,y_2)$ and $x_1,x_2,y_1,y_2\in \F_2^n$
\end{definition}
The following fact is easily verified:
\begin{proposition} For $x,y\in \F_2^{2n}$,
  \begin{equation}
    W_xW_y=(-1)^{[x,y]}W_yW_x.
  \end{equation}
\end{proposition}
\begin{definition}
An isotropic subspace of $\F_2^{2n}$ is a subspace $V$ of $\F_2^{2n}$ such that for all $x,y\in V$, $[x,y]=0$. If in addition $|V|=2^n$, we call the subspace \emph{Lagrangian}. 
\end{definition}
Lagrangian subspaces (which can be directly identified with the stabilizer groups of stabilizer states) are of importance to us, because they allow us to lower bound the stabilizer fidelity:
\begin{restatable}[Proof of Theorem 3.3 of \cite{gross2021schur},
  Corollary 7.4 of \cite{grewal2024improved}]
  {proposition}{propmainlowerbound}
  \label{prop:main-lowerbound}
  Let $V$ be a Lagrangian subspace of $\F_2^{2n}$. Then, 
    \begin{equation}
      \mathcal{F}_S\p{\ket{\psi}}\geq \sum_{x\in V}p_{\psi}(x).
    \end{equation}
    where $p_{\psi}(x) =2^{-n}|\langle \psi|W_x|\psi\rangle|^2$ is the characteristic distribution of $\ket{\psi}$ (see below).
\end{restatable}

For completeness, we give a proof of this proposition in \Cref{appendix:prop:main-lowerbound}.

\subsection{Bell difference sampling}
In this section we briefly discuss Bell difference sampling. The key objects are the \emph{characteristic} and \emph{Weyl} distributions of a state $\ket{\psi}$:
\begin{definition}[Characteristic and Weyl distributions]
We define the characteristic distribution of $\ket{\psi}$ as 
\begin{equation}
p_{\psi}(x)=2^{-n}|\langle \psi|W_x|\psi\rangle|^2.
\end{equation}
Furthermore, we define the Weyl distribution $q_\psi$ as
\begin{equation}
q_{\psi}(x)= 4^n(p_\psi*p_\psi)(x)=\sum_{y\in \F_2^{2n}} p_\psi(y)p_\psi(x+y).
\end{equation}
\end{definition}
We gather some properties of these distributions in the following proposition. Two of these are straightforward, and the third is a core result of \cite{arunachalam2024tolerant} (lemma 3.3):
\begin{proposition}[Properties of $p_\psi$]
    \label{prop:property-ppsi}
The characteristic distribution $p_{\psi}$ has the following properties:
    \begin{itemize}
        \item $\sum_{x\in \F_2^{2n}}p_{\psi}(x)=1$,
        \item For all $x\in \F_2^{2n}$, $p_{\psi}(x)\leq 2^{-n}$,
        \item $\E_{x\sim p_\psi}[2^np_{\psi}(x)]\geq \E_{x\sim q_\psi}[2^np_{\psi}(x)]\geq \p{\E_{x\sim p_\psi}[2^np_{\psi}(x)]}^2$.
    \end{itemize}
\end{proposition}

Finally we briefly discuss Bell difference sampling. Bell difference sampling is a surprisingly straightforward procedure that allows us to estimate $\E_{x\sim q_\psi}[|\langle\psi|W_x|\psi\rangle|^2]$ in the following way \cite{gross2021schur}:

\begin{enumerate}
    \item Measure $|\psi\rangle^{\otimes 2}$ in the Bell basis twice. This gives us two independent samples $x,y$ from $p_\psi$. We set $a=x+y$.
    \item Measure the Weyl operator $W_a$ twice. We accept the sample $a$ if the measurements agree, and reject otherwise.
\end{enumerate}
Repeating this many times and calculating the average number of accepts we can compute the expectation value $\E_{x\sim q_\psi}[|\langle\psi|W_x|\psi\rangle|^2]$. 
In particular (as can be seen in the proof of Theorem 3.3 in \cite{gross2021schur}), the probability of acceptance is equal to
\begin{equation}\label{paccept}
  p_{\text{accept}}=\frac{1}{2}\sum_a q_{\psi}(a)\p{1+|\langle \psi|W_a|\psi\rangle|^2}=\frac{1}{2}+\frac{1}{2}\E_{x\sim q_\psi}\lrb{\abs*{\bra{\psi}W_x\ket{\psi}}^2}.
\end{equation}
This expectation value is the core quantity in stabilizer testing. It provides an upper bound for the stabilizer fidelity as follows:
\begin{fact}[Lemma 3.1 of \cite{grewal2022low}]
  \label{fact:main-upperbound}
  Let $\ket{\psi}$ be an $n$-qubit pure quantum state. Then,
  \begin{equation}
    \mathcal{F}_S\p{\ket{\psi}} \le 
      \p{\E_{x\sim q_\psi}
        \lrb{\abs*{\bra{\psi}W_x\ket{\psi}}^2}}^{1/6}.
  \end{equation}
\end{fact}
This upper bound is enough to prove that Bell difference sampling provides a one-sided stabilizer tester. The rest of this note is concerned with providing a corresponding lower bound.

\section{A generalized uncertainty relation}

In this section we introduce a generalized uncertainty relation, which we will use to prove our main result. 
Roughly speaking, we aim to generalize the following folklore relation for sets of mutually anti-commuting Weyl operators to arbitrary sets of Weyl operators:

\begin{fact}[Folklore, Lemma 4.23 of \cite{arunachalam2024tolerant}]\label{fact:ur0}
  Let $\lrs{A_i}_{i=1}^M$ be a set of mutually anti-commuting
    Weyl operators.
  For any pure state $\psi=\ketbra{\psi}$, we have that  
    $\sum_{j=1}^M\Tr\p{\psi A_j}^2\le 1$.
\end{fact}
Let us now construct a generalization. Given an arbitrary set $A$ of Weyl operators, define $\Gamma_A$ as 
  the \emph{anti-commutation graph} of $A$. 
To be more specific, the vertex set of $\Gamma_A$ is exactly
  $A$, and for any $A_i,A_j\in A$, they are connected by an 
  edge if and only if $A_i$ and $A_j$ anti-commute.\\
  
Taking a slight detour, we can also define the set of ``normalized Hamiltonians'' associated with a set of Weyl operators $A$. We will be interested in the operator norm of these Hamiltonians, maximized over all possible Hamiltonians. For a set $A = \lrs{A_i}_{i=1}^M$ of Weyl operators, we define 
\begin{equation}
  \Psi_0\p{A} = \max_{a\in \R^M}
    \lrs{
      \norm*{H_A\p{a}^2}=\lambda_{\text{max}}\p{H_A(a)^2}
      \ \Bigg\vert\ 
      H_A\p{a}=\sum_{i=1}^M a_iA_i, 
      \norm{a}_2=\p{\sum_{i=1}^Ma_i^2}^{\f{1}{2}}=1
    }.
\end{equation}
Our key observation is as follows:
\begin{lemma}\label{lemma:ur-obversation}
  Let $\lrs{A_i}_{i=1}^M$ be a set of Weyl operators.
  For any pure state $\psi=\ketbra{\psi}$, we have that
  \begin{equation}
    \sum_{i=1}^M\Tr\p{\psi A_i}^2\le \Psi_0\p{A}.
  \end{equation}
\end{lemma}
\begin{proof}
  Consider the vector $w\in \R^{M}$ where 
    $w_i := \Tr \p{\psi A_i} = \bra{\psi}A_i\ket{\psi}$. 
  Notice that
  \begin{equation}
    \sum_{i=1}^M \Tr\p{\psi A_i}^2 
      = \sum_{i=1}^M \bra{\psi}A_i\ketbra{\psi}A_i\ket{\psi}
      = \bra{\psi}\sum_{i=1}^M w_i A_i\ket{\psi} 
      \le \norm*{H_A\p{w}},
  \end{equation}
  where $H_A\p{w}=\sum_{i=1}^{M}w_iA_i$.
  Since $\sum_{i=1}^M \Tr\p{\psi A_i}^2 = \norm{w}_2^2$ and
    $\norm*{H_A\p{w}}=\norm*{w}_2 \norm*{H_A\p{w/\norm*{w}_2}}$,
    we conclude that
  \begin{equation}
    \sum_{i=1}^M \Tr\p{\psi A_i}^2
      = \norm{w}_2^2
      \le \norm*{w}_2 \norm*{H_A\p{w/\norm*{w}_2}}
      \le \norm*{w}_2 \sqrt{\Psi_0\p{A}},
  \end{equation}
  since $\norm*{w/\norm*{w}_2}_2=1$.
  We can rearrange and square the above to obtain the 
    lemma statement. 
\end{proof}
We note that this quantity was also studied by \cite{anshuetz2024strongly} under the name "commutation index". The uncertainty relation we are about to show was also discovered independently (and earlier) in \cite{de2023uncertainty}. 
Next, we want to relate the quantity $\Psi_0\p{A}$ to a property of the anti-commutation graph~$\Gamma_A$. To do this, we leverage a result by Hastings \& O'Donnell \cite{hastings2022optimizing}, derived in the context of fermionic Hamiltonian optimization. Concretely, we need the following definition and proposition.

\begin{definition}[Definition 4.6 of \cite{hastings2022optimizing}]
  For a graph $\Gamma = \p{[M], E}$, we define
  \begin{equation}
    \Psi\p{\Gamma} = \max_{a\in \R^M}
      \lrs{
        \text{Opt}\p{l^2} 
        \ \Bigg\vert\ 
        l = \sum_{i=1}^{M} a_i\chi_i\in \mathfrak{C}\p{\Gamma},
        \norm{a}_2=\p{\sum_{i=1}^Ma_i^2}^{\f{1}{2}}=1
      }.
  \end{equation}
\end{definition}
Here, $\mathfrak{C}\p{\Gamma}$ is any matrix representation of the anti-commutation relations encoded by the graph~$\Gamma$. Note that in particular, any set of Weyl operators $A$ is a representation of its anti-commutation graph $\Gamma_A$ (this is discussed in greater detail in Section 2.6 of \cite{hastings2022optimizing}). Furthermore, this implies that $\Psi_0\p{\Gamma}\leq\Psi\p{\Gamma}$ and the following graph-theoretic characterization is also provided:

\begin{proposition}[Proposition 4.8 from 
  \cite{hastings2022optimizing}]\label{prop:HOprop}
  For any graph $\Gamma$, we have $\Psi\p{\Gamma} \le 
    \vartheta\p{\Gamma}$.
\end{proposition}

With this we can show our generalized uncertainty relation. 

\begin{lemma}[Generalized uncertainty relation]\label{lem:gen_unc}
  Let $\lrs{A_i}_{i=1}^M$ be a set of Weyl operators with an associated anti-commutation graph $\Gamma_A$. 
  For any pure state $\psi=\ketbra{\psi}$, we have that
  \begin{equation}
    \sum_{i=1}^M\Tr\p{\psi A_i}^2 \le \Psi_0\p{A} 
      \le \Psi\p{\Gamma_A} \le \vartheta\p{\Gamma_A}.
  \end{equation}
\end{lemma}

\begin{proof}
  The first inequality is directly from 
    \Cref{lemma:ur-obversation}.
  By definition, we have $\Psi_0\p{\Gamma_A}\le 
    \Psi\p{\Gamma_A}$, which is the second inequality.
  The last inequality follows from \Cref{prop:HOprop}.
\end{proof}

\section{A tolerant testing algorithm}
In this section we prove our main theorem. We start from the core fact derived in \cite{arunachalam2024tolerant}, namely that if Bell difference sampling succeeds with high probability, there exists a subspace $V\subset \mathbb{F}_2^{2n}$ of Weyl operators  with high expected probability mass under the characteristic distribution $p_{\psi}$:
\begin{restatable}{theorem}{toltestmain}\label{thm:tol_test_main}
Let $\ket{\psi}$ be an $n$-qubit quantum state such that $\mathbb{E}_{x\sim q_{\psi}}\lrb{2^np_{\psi}(x)} \geq \gamma$ for $\gamma\in [0,1]$ and $2^n\ge \f{C''\ln\p{C'''/\gamma}}{\gamma^3}$ with $C''$ and $C'''$ some constants. 
Then there exists a subspace $V$ of $\mathbb{F}_{2}^{2n}$ such that:
\begin{equation}\label{eq:weight_low_bound}
\sum_{x\in V} \abs*{\bra{\psi}W_x\ket{\psi}}^2 \geq C_1 \gamma^{55} \abs*{V},
\end{equation}
and
\begin{equation}
\sum_{x\in V} \abs*{\bra{\psi}W_x\ket{\psi}}^2 \geq C_2 \gamma^{57} 2^n,
\end{equation}
with $C_1$ and $C_2$ some constants.
\end{restatable}
We prove this theorem in the appendix. A version of this theorem, that does not include the mild $n$ dependence, can also be proven by following the logic of \cite{arunachalam2024tolerant} through their Section $4$ towards Corollary 4.11, and Claims 4.12 and 4.13 (together with Lemma 3.3), and taking care to tally up all powers of $\gamma$ (which is not done explicitly in \cite{arunachalam2024tolerant}). This yields a similar theorem but with powers of $\gamma$ equal to $406$ and $413$ instead of $55$ and $57$. \\

Our goal will now be to find an isotropic subspace $V_0$ of the subspace $V$ that also has a large probability mass.
We will do this by providing an upper bound and a lower bound on the Lov{\'a}sz theta of the anti-commutation graph $\Gamma_V$ of (Weyl operators associated to) the subspace $V$. To do this we note that for the Weyl operators associated to a subspace $V$ there exist integers $m$, $k$ (with $k+m \leq n$) and an $n$-qubit Clifford unitary $U$ such that (abusing notation somewhat):
\begin{equation}
    UVU^\dagger = \langle Z_1,X_1,\dots, Z_k,X_k,Z_{k+1},\dots,Z_{k+m}\rangle \defeq W.
\end{equation}
where $P_i$ is a Pauli operator on the $i$'th qubit. 
With this structure, and the uncertainty relation we derived above we can control $\vartheta(\Gamma_W) = \vartheta(\Gamma_V)$ from both above and below.
\begin{lemma}\label{lemma:ur-final}
  Let $W$, $k$, $m$ be defined as above, with
    associated anti-commutation graph $\Gamma_W$.
  For any pure state $\psi=\ketbra{\psi}$, we have that
  \begin{equation}
    C_1 \gamma^{55} |W| \leq \sum_{x \in W}\Tr\p{\psi W_x}^2 \le \vartheta\p{\Gamma_W}
      \le \f{\abs*{W}}{2^k} = 2^{k+m}.
  \end{equation}
\end{lemma}

\begin{proof}
The lower bound is given by \Cref{lem:gen_unc} and \Cref{thm:tol_test_main}.  It is thus sufficient to show $\vartheta\p{\Gamma_W} 
    \le \f{\abs*{W}}{2^k}$.
  It is straightforward to see that
  \begin{equation}
    \Gamma_W \supseteq \Gamma_{\mathcal{P}_{k}} 
      \boxtimes \overline{K_{2^m}},
  \end{equation}
  where $\Gamma_{\mathcal{P}_{k}}$ is the anti-commutation graph
    of the (Weyl operators in the) $k$-fold Pauli group, and $\overline{K_{2^m}}$ is the 
    complement graph of the complete graph with $2^m$ vertices.
  According to properties of the Lov{\'a}sz theta 
    (\Cref{fact:theta-mono}, \Cref{fact:theta-product}), we have
  \begin{equation}
    \vartheta\p{\Gamma_W} 
      \le \vartheta\p{\Gamma_{\mathcal{P}_{k}}
        \boxtimes \overline{K_{2^m}}}
      = \vartheta\p{\Gamma_{\mathcal{P}_{k}}}
        \cdot \vartheta\p{\overline{K_{2^m}}}.
  \end{equation}
  Note that $\Gamma_{\mathcal{P}_{k}}$ is actually the 
    complement graph of the \emph{symplectic graph} 
    $\textrm{Sp}\p{2k,2}$ defined in~\cite{sason2023observations} 
    with one extra isolated vertex 
    (corresponding to the identity operator).
  According to \Cref{fact:theta-sum} we have 
  $\vartheta\p{\Gamma_{\mathcal{P}_{k}}}
    =\vartheta\p{\overline{\textrm{Sp}(2k,2)}}+1$.
  With $\vartheta\p{\overline{\textrm{Sp}(2k,2)}} = 2^k-1$ (from 
    Theorem $3.29$ of \cite{sason2023observations}),
    $\vartheta\p{\overline{K_{2^m}}} = 2^m$ (directly from the definition of $\vartheta$), and $\abs*{W}=2^{2k+m}$, we conclude that
  \begin{equation}
    \vartheta\p{\Gamma_W}\le 2^{k+m} = \f{\abs*{W}}{2^k}.
  \end{equation}
\end{proof}

As a conclusion of this section, We have the following 
  corollary, which will be used to prove a polynomial gap
  tolerant testing theorem for stabilizer states.

\begin{corollary}\label{coroll:main-lowerbound}
Let $\ket{\psi}$ be a pure state, and let $V$ be the subspace given in \Cref{thm:tol_test_main}, with associated anti-commutation graph $\Gamma_V$.
There exists an isotropic subspace $V_0\subseteq V$, 
    such that 
  \begin{equation}
    \sum_{x\in V_0} \Tr\p{\psi W_x}^2\ge 
      \f{C_1 \gamma^{55}}{1+C_1\gamma^{55}} 
        \sum_{x\in V} \Tr\p{\psi W_x}^2.
  \end{equation}
\end{corollary}

\begin{proof}
  By \Cref{lemma:ur-final} and \Cref{eq:weight_low_bound}, we have that
  \begin{equation}
    C_1 \gamma^{55} \cdot\abs*{V}\le \sum_{x\in V} \Tr\p{\psi W_x}^2
      \le \vartheta\p{\Gamma_V} \le \f{\abs*{V}}{2^k},
  \end{equation}
  which implies $2^k\le \f{1}{C_1}\gamma^{-55}$.
  As a direct consequence of Lemma 12 from
    \cite{leung2012entanglement}, we can thus cover $V$ with
    $2^k+1$ isotropic subspaces of $V$.
  Therefore there must exist a specific isotropic $V_0\subseteq V$,
    such that
  \begin{equation}
    \sum_{P\in V_0} \Tr\p{\psi P}^2
      \ge \f{1}{2^k+1} \sum_{P\in V} \Tr\p{\psi P}^2
      \ge \f{C_1 \gamma^{55}}{1+C_1 \gamma^{55}} 
        \sum_{P\in V} \Tr\p{\psi P}^2.
  \end{equation}
\end{proof}
Note that the existence of $V_0$ lower bounds the stabilizer fidelity of $\ket{\psi}$, since we can always extend this isotropic subspace to a Lagrangian subspace, which gives rise to a stabilizer state. This leads to the main theorem. 
\begin{theorem}\label{fidelitysandwich}
  Let $\ket{\psi}$ be a state and let $\gamma = \mathbb{E}_{x\sim
    q_{\psi}}\lrb{2^np_{\psi}(x)}$. Then there exists some constant  $C$, such that
  \begin{equation}
    C\gamma^{112} \le \mathcal{F}_S(\ket{\psi}) \le \gamma^{\f{1}{6}}.
  \end{equation}
\end{theorem}

\begin{proof}
  The second inequality is exactly \Cref{fact:main-upperbound}.
  We will prove the first inequality. We begin by padding $\ket{\psi}$ with $\ket{0}$ states such that $2^n\ge \f{C''\ln\p{C'''/\gamma}}{\gamma^3}$. Note that this leaves the stabilizer fidelity unchanged (this follows from the contraction identity in App. G of \cite{hayden2016holographic}). Also note that $\gamma\geq \mathcal{F}_S(\ket{\psi})^6 \geq 2^{-6n}$, so this never requires more that $cn$ qubits to do\footnote{This padding procedure is a technical trick to overcome the large-$n$ requirement in the proof of \Cref{thm:tol_test_main}. We believe this is a proof artifact. In fact it can be removed by using the version of \Cref{thm:tol_test_main} from \cite{arunachalam2024tolerant}, at the cost of a substantially higher power of $\gamma$.}.
  From \Cref{thm:tol_test_main} and \Cref{coroll:main-lowerbound},
    there exists a subspace $V$, an isotropic subspace $V_0$ and 
    constants $C_1$, $C_2$, such that:
  \begin{equation}
    \sum_{x\in V_0} \Tr\p{\psi W_x}^2
      \ge 
        \f{C_1 \gamma^{55}}{1+C_1 \gamma^{55}} 
          \sum_{x\in V} \Tr\p{\psi W_x}^2
      \ge
        \f{C_1\gamma^{55}}{1+C_1\gamma^{55}} 
          C_2 \gamma^{57} 2^n 
      =
        \f{C_1C_2\gamma^{112}}{1+C_1\gamma^{55}} 2^n.
  \end{equation}
  We can extend $V_0$ to a Lagrangian subspace $V_0^*$ such that
  \begin{equation}
    \sum_{x\in V_0^*} p_{\psi}(x) 
      = \sum_{x\in V_0^*} \f{\Tr\p{\psi W_x}^2}{2^n}
      \ge \f{C_1C_2\gamma^{112}}{1+C_1\gamma^{55}}
  \end{equation}
  Following \Cref{prop:main-lowerbound}, we conclude that
  \begin{equation}
    \mathcal{F}_S(\ket{\psi}) 
      \ge \f{C_1C_2\gamma^{112}}{1+C_1\gamma^{55}} 
      \ge \f{C_1C_2\gamma^{112}}{1+C_1} 
      = C\gamma^{112},
 \end{equation}
  where $C=\f{C_1C_2}{1+C_1}$.
\end{proof}
From this we can obtain a tolerant testing algorithm. The argument is standard, but we include it for completeness, and to get an explicit value for the polynomial degree in the gap. For the proof of \Cref{theoremone} we need a standard Chernoff bound: 
\begin{lemma}[Chernoff-Hoeffding]\label{chernoffhoeffding}
    Let $X_1,\dots,X_m$ be independent random variables such that $a\leq X_i\leq b$ almost surely for all $i$. Then,
    \[\PP\left(\left|\frac{1}{m}\sum_{i=1}^mX_i-\E\left[\frac{1}{m}\sum_{i=1}^mX_i\right]\right|\geq \delta\right)\leq 2e^{-\frac{2\delta^2m}{(b-a)^2}}\]
\end{lemma}
\theoremone*
\begin{proof}
  Let $\gamma=\E_{x\sim q_\psi}\lrb{|\langle \psi|W_x|\psi\rangle|^2}$. We choose $D_1,D_2$ such that
    \begin{equation}
    \begin{split}
        D_1:=&\ \epsilon_1^6,\\
        D_2:=&\left(\frac{\epsilon_2}{C}\right)^{1/112},
    \end{split}
    \end{equation}
    where $C$ is the constant from \Cref{fidelitysandwich}.
    We define $D:=\frac{D_1+D_2}{2}$.
    The algorithm now does the following:
    We perform Bell sampling $m$ times, obtaining samples $x_1,\dots,x_m$ to compute an estimate $\overline{\gamma}$ of the expectation value.
    The $x_1,\dots,x_m$ are Bernoulli distributed with mean $p_{\text{accept}}=\frac{1}{2}+\frac{1}{2}\gamma$ (see \Cref{paccept}). We will set our estimate $\overline{\gamma}$ to be
    \begin{equation}
    \overline{\gamma}=\frac{1}{m}\sum_{i=1}^m\left(2x_i-1\right).
    \end{equation}
    Note that $\E[\overline{\gamma}]=\gamma$. We make the following decisions based on this data:
    \begin{itemize}
        \item If $\overline{\gamma}\geq D$, we output $\mathcal{F}_S(\ket{\psi})\geq \epsilon_1$,
        \item If $\overline{\gamma}< D$, we output $\mathcal{F}_S(\ket{\psi})\leq \epsilon_2$.
    \end{itemize}
        From the promise we know that either $\mathcal{F}_S(|\psi\rangle)\geq \epsilon_1$, or that $\mathcal{F}_S(|\psi\rangle)\leq \epsilon_2$.
    Suppose that (1) $\mathcal{F}_S(|\psi\rangle)\geq \epsilon_1$, then by \Cref{fidelitysandwich}:
        \begin{equation}
        \gamma^{1/6}\geq \mathcal{F}_S(|\psi\rangle)\geq \epsilon_1.
    \end{equation}
    Thus $\gamma \geq D_1$.
    On the other hand, if (2) $\mathcal{F}_S(|\psi\rangle)\leq \epsilon_2$, then, again by \Cref{fidelitysandwich},
        \begin{equation}
        C\gamma^{112}\leq \mathcal{F}_S(|\psi\rangle)\leq \epsilon_2.
    \end{equation}
    Thus $\gamma \leq D_2$.
    This tells us that 
    \begin{equation}
        \gamma\leq D_2\text{ or }D_1\leq \gamma.
    \end{equation}
    Now suppose that $\mathcal{F}_S(\ket{\psi})\geq \epsilon_1$, but $\overline{\gamma}<D$. Then the true value of $\gamma$ is greater than or equal to $D_1$. Therefore, the difference $|\overline{\gamma}-\gamma|\geq \alpha$, where we define $\alpha=\frac{D_1-D_2}{3}$.
    We set $C'=\frac{C}{2^{112}}$. By assumption we have that $\epsilon_2\leq \frac{C}{2^{112}}\epsilon_1^{672} $. Equivalently, $\left(\frac{\epsilon_2}{C}\right)^{1/112}\leq\frac{1}{2}\epsilon_1^6 $, so $D_2\leq \frac{1}{2}D_1$. Therefore, it holds that $D_1-D_2\geq \frac{1}{2} D_1$, so $\alpha \geq \frac{1}{6}D_1$.
    Similarly, suppose that $\mathcal{F}_S(\ket{\psi})\leq \epsilon_2$, but $\overline{\gamma}\geq D$. Then the true value of $\gamma$ is less than or equal to $D_2$. Therefore, the difference $|\overline{\gamma}-\gamma|\geq \alpha$.
    So using the a Chernoff-Hoeffding bound (\Cref{chernoffhoeffding}) with $a=-1,b=1$, we can upper bound the probability that the algorithm makes a mistake by
    \begin{equation}
        \PP\left(|\overline{\gamma}-\gamma|\geq \alpha\right)\leq 2e^{-\frac{1}{2}\alpha^2 m}\leq2e^{-\frac{1}{72}\epsilon_1^{12} m} .
    \end{equation}
    Now assuming that $m\geq \frac{\log(3)72}{\epsilon_1^{12}}$, this probability is lower than $1/3$. 
\end{proof}
\paragraph{Acknowledgements}
We would like to thank the Qusoft testing \& learning reading group: Jop Bri{\"e}t, Amira Abbas, Yanlin Chen, Niels Neumann, Davi Castro-Silva \& Jeroen Zuiddam, without whom this paper would not have been written. We also thank Krystal Guo and M\={a}ris Ozols for teaching us about symplectic graphs. JH would like to thank Johannes Jakob Meyer for pointing out reference \cite{anshuetz2024strongly}, and Carlos de Gois for alerting us to his work \cite{de2023uncertainty} after we put out a first version of this paper. JH acknowledges funding from the Dutch Research Council (NWO) through a Veni grant (grant No.VI.Veni.222.331) and the Quantum Software Consortium (NWO Zwaartekracht Grant No.024.003.037). ZB is supported by the Dutch Ministry of Economic Affairs and Climate Policy (EZK), as part of the Quantum Delta NL programme (KAT2).

\bibliographystyle{alpha}
\bibliography{paper}

\appendix

\section{An improved bound on the polynomial degree}

In this section we present an improvement to the degree of the polynomial gap of the tester, beyond what can be achieved by directly relying on the work of \cite{arunachalam2024tolerant} (summarized below \Cref{thm:tol_test_main}). We will almost entirely retrace their argument but obtain various parametric improvements and simplifications. The main bottleneck in obtaining a lower degree polynomial relationship is the black box use of theorems from additive combinatorics. It is unclear how this can be fully avoided.

\subsection{Obtaining a large nearly linear subset}
The main contribution of this appendix is a sharply refined version of Theorem 4.5 of \cite{arunachalam2024tolerant}, which shows the existence of a large set $S \subseteq \mathbb{F}_2^{2n}$ that is nearly linear (in a precise probabilistic sense) and has large point-wise probability mass under the characteristic distribution $p_{\psi}$. We believe that our version of this result is optimal (at least with respect to its scaling with $\gamma$). We note that the techniques in this theorem could also be used to parametrically improve the landmark result on low degree testing from Samorodnitsky\cite{samorodnitsky2007low}. We have the following:
\begin{theorem}\label{thm:refined-Theorem4.5}
Let $|\psi\rangle$ be an $n$-qubit quantum state. 
Then if $\E_{x\sim q_{\psi}}\lrb{\abs*{\langle 
  \psi|W_x|\psi\rangle}^2}\geq \gamma$, 
  and $2^n\ge \f{C''\ln\p{C'''/\gamma}}{\gamma^3}$ with $C''$ and $C''$
  some constants,
  there exists a set $S\subseteq \F_2^{2n}$ such that
\begin{enumerate}
    \item $\abs*{S}\ge \frac{\gamma}{2}2^n$,
    \item For all $x\in S$, $2^np_{\psi}(x)\geq \f{\gamma}{4}$,
    \item $\PP_{s,s'\in S}\lrb{s+s'\in S}\geq \f{\gamma}{6}$.
\end{enumerate}
\end{theorem}
\begin{proof}
We begin by defining the set:
\begin{equation}
    X=\lrs{x\in \F_2^{2n}:2^np_{\psi}(x)\geq \frac{\gamma}{4}}.
\end{equation}
We prove that this set is large, and then construct $S$ as a random subset of $X$. 
It is clear that $\f{\gamma}{4}\abs*{X}\le \sum_{x\in X}2^np_\psi(x)\le \sum_{x\in \Fnn}2^np_\psi(x) = 2^n$. This implies that $\abs*{X}\le \f{4}{\gamma}2^n$.
Furthermore, we can lower bound $|X|$ using Markov's inequality:
\begin{equation}
  \gamma 
    \leq \E_{x\sim q_\psi}\lrb{2^np_{\psi}(x)}
    \leq \PP_{x\sim q_\psi}\lrb{2^np_\psi(x)\geq \f{\gamma}{4}}+\f{\gamma}{4}.
\end{equation}
Since $q_{\psi}(x)\leq 2^{-n}$ for all $x$, we have
\begin{equation}
  \PP_{x\sim q_\psi}\lrb{2^np_\psi(x) 
    \geq \f{\gamma}{4}}
    = \sum_{x\in X}q_{\psi}(x)\leq |X|2^{-n}.
\end{equation}
Combining the above we obtain a lower bound:
\begin{equation}
  \abs*{X} \ge \f{3}{4}\gamma \cdot 2^n.
\end{equation}
Next we construct the set $S$. This is done through a sampling process on $X$. We take every $x \in X$, and independently randomly decide to include $x\in S$ with probability $2^np_{\psi}(x)$. We will now prove that $S$ satisfies the three properties of the theorem statement with non-zero probability. Note that the second property is automatically satisfied, since $S\subseteq X$. \\

The rest of the argument is probabilistic (over the random choice of $S$). We begin by calculating the expected size of $S$. It's easy to see that 
\begin{equation}
\E_S\lrb{\abs*{S}}=\sum_{x\in X}2^np_\psi\p{x}
  =2^n\PP_{x\sim p_\psi}\lrb{x\in X}=2^n\PP_{x\sim p_\psi}
  \lrb{2^np_\psi(x)\ge \f{\gamma}{4}}.
  \end{equation}
Furthermore, from \Cref{prop:property-ppsi} we know that 
  $\E_{x\sim q_\psi}\lrb{2^np_{\psi}(x)} \le \E_{x\sim p_\psi}\lrb{2^np_{\psi}(x)}$.
According to Markov's inequality, we have
\begin{equation}
  \gamma
    \le \E_{x\sim q_\psi}\lrb{2^np_{\psi}(x)}
    \le \E_{x\sim p_\psi}\lrb{2^np_{\psi}(x)}
    \le \PP_{x\sim p_\psi}\lrb{2^np_\psi(x)\geq \f{\gamma}{4}}+\f{\gamma}{4}.
\end{equation}
Therefore $\E_S\lrb{\abs*{S}}\ge \f{3}{4}\gamma \cdot 2^n$.
Next we compute the expected ``linearity" of $S$. In this argument we will need an upper bound on the size of $S$, which we will indirectly obtain from its expected size. Let $\lambda$ be a constant, which we will choose later, and define $A$ to be the event $\{\abs*{S} < \p{1+\lambda}\E_S\lrb{\abs*{S}}\}$.\\

\noindent Now we can compute the linearity. Define 
\begin{equation}
L(S)=\PP_{s,s'\in S}\lrb{s+s'\in S}.
\end{equation}
We calculate the expected linearity $\E_S\lrb{L(S)}$ as follows:
\begin{align}
  \E_S\lrb{L(S)} &= \E_S\lrb{\abs*{S}^{-2}\sum_{x,y\in S}\1[x+y\in S]} \\ 
    &= \E_S\lrb{\abs*{S}^{-2}\sum_{x,y\in X}\1[x\in S]\1[y\in S]\1[x+y\in S]} \\ 
    &\ge \E_S\lrb{\abs*{S}^{-2}\sum_{x,y\in X} 
      \1[x\in S]\1[y\in S]\1[x+y\in S] \1[A]} \\
    &\ge \f{2^{-2n}}{\p{1+\lambda}^2} 
      \E_S\lrb{\sum_{x,y\in X} \1[x\in S]\1[y\in S]\1[x+y\in S] \1[A]} \label{eq:S}
\end{align}
where the last inequality follows from $\abs*{S}\le \p{1+\lambda}\E_S\lrb{\abs*{S}}$
  and $\E_S\lrb{\abs*{S}}=2^n\PP_{x\sim p_\psi}\lrb{x\in X}\le 2^n$. We can continue this calculation:
\begin{align}
 &\f{2^{-2n}}{\p{1+\lambda}^2} 
      \E_S\lrb{\sum_{x,y\in X} \1[x\in S]\1[y\in S]\1[x+y\in S] \1[A]} \notag\\
    &\hspace{8em}= \f{2^{-2n}}{\p{1+\lambda}^2} \sum_{x,y\in X}
      \E_S\lrb{\1[x\in S]\1[y\in S]\1[x+y\in S]\1[A]}\\
    &\hspace{8em}\ge \f{2^{-2n}}{\p{1+\lambda}^2} \sum_{x,y\in X} \p{
      \E_S\lrb{\1[x\in S]\1[y\in S]\1[x+y\in S]} - \PP_S\lrb{\bar{A}} }\\
    &\hspace{8em}=  \f{2^{-2n}}{\p{1+\lambda}^2} \p{\sum_{x,y\in X}
        \E_S\lrb{\1[x\in S]\1[y\in S]\1[x+y\in S]}
        - \abs*{X}^2\PP_S\lrb{\bar{A}}\label{eq:last}
      }.
\end{align}
Let us now compute $\PP_S\lrb{\bar{A}}$.
From  Hoeffding's inequality, and the bounds 
  $\E_S\lrb{\abs*{S}}\ge \f{3}{4}\gamma \cdot 2^n$ and $\abs*{X}\le \f{4}{\gamma}2^n$,
  we have that
\begin{align}
  \PP_S\lrb{\bar{A}}
    = \PP_S\lrb{\abs{S} \ge \E_S\lrb{\abs*{S}} + \lambda \E_S\lrb{\abs*{S}}}
    \le \exp{-\f{2\lambda^2 \p{\E_S\lrb{\abs*{S}}}^2}{\abs*{X}}}
    &\le \exp{-\f{2\lambda^2 \f{9}{16}\gamma^2 \cdot 2^{2n}}{\f{4}{\gamma}2^n}} \\
    &= \exp{-9\lambda^2\gamma^3 2^{n-5}}.
\end{align}

Next we address the other term in \Cref{eq:last}:
\begin{align}
  \sum_{x,y\in X} \E_S\big[\1[x\in S]\1[y\in S]\1[x+y\in S]\big]
    &= \sum_{\substack{x,y\in X;\\ x+y\in X}}
      \E_S\big[\1[x\in S]\1[y\in S]\1[x+y\in S]\big] \\
    &\ge \sum_{\substack{x,y\in X;\\ x+y\in X}}2^{3n}p_\psi(x)p_\psi(y)p_\psi(x+y) \\
    &= 2^{2n} \sum_{\substack{x,y\in X;\\ x+y\in X}}2^{n}p_\psi(x)p_\psi(y)p_\psi(x+y).\label{eq:sumX}
\end{align}
To evaluate this last sum, we relate it to the related sum over $\mathbb{F}_2^{2n}$ (as opposed to $X$).  Note that we have 
\begin{equation}
\sum_{x,y\in \Fnn}2^n p_\psi(x)p_\psi(y)p_\psi(x+y)=\E_{x\sim q_{\psi}}\lrb{\abs*{\langle \psi|W_x|\psi\rangle}^2}\geq \gamma.
\end{equation}
Relating this sum to the summation in \Cref{eq:sumX} requires bounding three residual sums:
\begin{align}
  \sum_{x,y\in \Fnn; x\notin X} 2^n p_\psi(x)p_\psi(y)p_\psi(x+y) 
    &\le \f{\gamma}{4} \sum_{x,y\in \Fnn} p_\psi(y)p_\psi(x+y) = \f{\gamma}{4}, \\
  \sum_{x,y\in \Fnn; y\notin X} 2^n p_\psi(x)p_\psi(y)p_\psi(x+y) 
    &\le \f{\gamma}{4} \sum_{x,y\in \Fnn} p_\psi(x)p_\psi(x+y) = \f{\gamma}{4},\\
  \sum_{x,y\in \Fnn; x+y\notin X} 2^n p_\psi(x)p_\psi(y)p_\psi(x+y) 
    &\le \f{\gamma}{4} \sum_{x,y\in \Fnn} p_\psi(x)p_\psi(y) = \f{\gamma}{4},
\end{align}
Combining these and working out we obtain the lower bound:
\begin{equation}
  \E_S\lrb{L(S)} 
    \ge \f{1}{\p{1+\lambda}^2} \p{\gamma - \f{3}{4}\gamma - 
      2^{-2n}\abs*{X}^2e^{-9\lambda^2 \gamma^3 2^{n-5}}}
    \ge \f{1}{\p{1+\lambda}^2} \p{\f{1}{4}\gamma - 
      \f{16}{\gamma^2}e^{-9\lambda^2 \gamma^3 2^{n-5}}}.
\end{equation}
Fix $\lambda = 0.01$, for $n$ such that $2^n\ge \f{320000}{9}\cdot \f{3\ln\p{1/\gamma}+\ln 1600}{\gamma^3}
  = \f{C''\ln\p{C'''/\gamma}}{\gamma^3}$ with $C'$ and $C''$ some 
  constants, we have $\E_S\lrb{L(S)} \ge \f{\gamma}{5}$.
Using the Chernoff bound, Markov's inequality, and the union bound, we conclude that there exists $S$,
  such that $\abs*{S}\ge \f{\gamma}{2}\cdot 2^n$ and 
  $\PP_{s,s'\in S}\lrb{s+s'\in S}\geq \f{\gamma}{6}$.
\end{proof}

\subsection{Proof of \texorpdfstring{\Cref{thm:tol_test_main}}{theorem 19}}

Using \Cref{thm:refined-Theorem4.5} we can prove \Cref{thm:tol_test_main}. This argument is exactly that of \cite{arunachalam2024tolerant}, chaining together the Balog-Szemeredi-Gowers and polynomial Freiman-Rusza theorems to obtain a subspace $V\subseteq \mathbb{F}_2^{2n}$ with appropriate parameters. We include it here in the interest of presenting a self-contained narrative. 

\toltestmain*

\begin{proof}
We use \Cref{thm:refined-Theorem4.5} to guarantee the existence of a set $S$ (with the properties given in \Cref{thm:refined-Theorem4.5}. As we disussed there, this set is nearly linear. We can apply the Balog-Szemeredi-Gowers theorem to $S$, to give us another set $S'$ which is of similar size and has \emph{small doubling} (meaning the set $S'+S' := \left\{s+\hat{s}\;:\: s,\hat{s}\in S'\right\}$ is not much bigger than $S'$. Concretely the BSG theorem (adapted from \cite{viola2011selected}) is stated as follows:
\begin{restatable}[Balog-Szemeredi-Gowers, modified form
  \cite{balog1994statistical, gowers2001new, JopLectureNote}]{theorem}{BSG}
  Let $S\subseteq \F_2^{2n}$ such that $\PP_{s,s'\in S}\lrb{s+s'\in S}\geq \epsilon$, 
  then there exists a subset $S'\subseteq S$, such that
  \begin{enumerate}
    \item $|S'|\geq \frac{\epsilon}{2\sqrt{2}}\abs*{S}$,
    \item $\frac{\abs*{S'+S'}}{\abs*{S'}}\leq 8\p{\f{1}{\vep}}^6$.
  \end{enumerate}
\end{restatable}
The proof for this form of BSG theorem is based on \cite{JopLectureNote}, is postponed to \Cref{app:BSG}.
We apply this theorem with $\epsilon=\frac{\gamma}{6}$ to obtain a set $S'$ with $|S'|\geq \frac{\gamma}{18}\abs*{S}$ and $\frac{\abs*{S'+S'}}{\abs*{S'}}\leq \p{\frac{36}{\gamma}}^6$.\\

The next step is to argue that a set of small doubling can be covered by a few translates of a subspace $V$. This is the content of the polynomial Freiman-Ruzsa theorem.
\begin{theorem}[Polynomial Freiman-Ruzsa Theorem
  \cite{gowers2023conjecture,gowers2024marton,liao2024improved}]
If a set $S'\subseteq \F_2^{2n}$ has doubling constant at most $K$, then there exists a subspace $V$ of $\F_2^{2n}$ with the follwing properties:
\begin{enumerate}
    \item $|V|\leq |S'|$,
    \item $S'$ is covered by at most $(2K)^9$ translates of $V$.
\end{enumerate}
\end{theorem}
We apply the Polynomial Freiman-Ruzsa theorem to $S'$, with $K=\left(\frac{36}{\gamma}\right)^6$, yielding a subspace $V\subseteq \mathbb{F}_2^{2n}$. Now by the pigeonhole principle, there must exist a translate of $V$, say $V+y$, that contains at least $\frac{1}{(2K)^8}|S'|$ members of $S'$. Each element $x\in S'$ has $p_{\psi}(x)\geq 2^{-n}\f{\gamma}{4}$, so 
\begin{equation}
  \frac{1}{|V+y|}\sum_{x\in V+y}p_{\psi}(x)\geq \frac{1}{|V|}\sum_{x\in (V+y)\cap S'}2^{-n}\gamma/4\geq \frac{1}{|V|}(2K)^{-9}|S'|2^{-n}\f{\gamma}{4}\geq (2K)^{-9}2^{-n}\f{\gamma}{4}
\end{equation}
So we have
\begin{equation}
\E_{x\in V+y}p_{\psi}(x)\geq (2K)^{-9}2^{-n}\f{\gamma}{4} = C_1\gamma^{55}2^{-n},
\end{equation}
for some constant $C_1$.
Through a clever little argument in \cite{arunachalam2024tolerant} (Corollary 4.11), we have that this value is even larger on $V$: 
\begin{equation}
\E_{x\in V}p_{\psi}(x)\geq\E_{x\in V+y}p_{\psi}(x)\geq C_1\gamma^{55}2^{-n}.
\end{equation}
If we rewrite this equation we get
\begin{equation}
    \sum_{x\in V} \abs*{\bra{\psi}W_x\ket{\psi}}^2 \geq C_1 \gamma^{55} \abs*{V}.
\end{equation}
Besides,
\begin{equation}
  \sum_{x\in V} \abs*{\bra{\psi}W_x\ket{\psi}}^2 
    \ge
      \p{2K}^{-8}\sum_{x\in S'} 
        \abs*{\bra{\psi}W_x\ket{\psi}}^2
    \ge
      \p{2K}^{-8} \f{\gamma}{4} \abs*{S'}.
\end{equation}
Note that $\abs*{S'}\ge \f{\eps}{3}\cdot \abs*{S} 
  \ge \f{\eps}{3} \cdot \f{\gamma}{2} 2^n$, 
  we conclude that
\begin{equation}
  \sum_{x\in V} \abs*{\bra{\psi}W_x\ket{\psi}}^2 
    \ge \p{2K}^{-9} \f{\gamma}{4}\cdot 
      \f{\eps}{3} \cdot \f{\gamma}{2} 2^n
    = C_2 \gamma^{57} 2^n.
\end{equation}

\end{proof}
\section{Proof of \Cref{prop:main-lowerbound}}\label{appendix:prop:main-lowerbound}
\propmainlowerbound*
\begin{proof}
Let $V$ be Lagrangian subspace of $\F_2^n$. Then let $\{|\phi_\alpha\rangle\}_{\alpha\in \F_2^n}$ be the basis of eigenstates of the matrices $\{W_x:x\in V\}$. For each $\alpha$, $|\phi_\alpha\rangle$ is a stabilizer state. 
Every $x\in V$ corresponds to a unique linear function $f_x:\F_2^n\rightarrow\F_2$ such that
\begin{equation}
    W_x=\sum_\alpha (-1)^{f_x(\alpha)}\ket{\phi_\alpha}\bra{\phi_\alpha}
\end{equation}
Now we have, setting $\psi=\ketbra{\psi}{\psi}$,
\begin{equation}
    \mathcal{F}_S(\ket{\psi})\geq \max_{\alpha} \bra{\phi_\alpha}\psi\ket{\phi_\alpha}=\max_{\alpha} \bra{\phi_\alpha}\psi\ket{\phi_\alpha}\sum_\beta\bra{\phi_\beta}\psi\ket{\phi_\beta}\geq\sum_\beta\bra{\phi_\beta}\psi\ket{\phi_\beta}^2.
\end{equation}
We show that $2^n\sum_\beta\ket{\phi_\beta}\bra{\phi_\beta}\psi\ket{\phi_\beta}\bra{\phi_\beta}=\sum_{x\in V}W_x\psi W_x^\dagger$ :
\begin{equation}
    \begin{split}
        \sum_{x\in V}W_x\phi W_x^\dagger
        =&\sum_{x\in V}\sum_{\alpha,\beta\in \F_2^n}(-1)^{f_x(\alpha)}\ket{\phi_\alpha}\bra{\phi_\alpha}\psi (-1)^{f_x(\beta)}\ket{\phi_\beta}\bra{\phi_\beta}\\
        =&\sum_{x\in V}\sum_{\alpha,\beta\in \F_2^n}(-1)^{f_x(\alpha+\beta)}\ket{\phi_\alpha}\bra{\phi_\alpha}\psi\ket{\phi_\beta}\bra{\phi_\beta}\\
        =&\sum_{\alpha,\beta\in \F_2^n}2^n\mathbbm{1}[\alpha=\beta]\ket{\phi_\alpha}\bra{\phi_\alpha}\psi\ket{\phi_\beta}\bra{\phi_\beta}\\
        =&\sum_{\beta\in \F_2^n}2^n\ket{\phi_\beta}\bra{\phi_\beta}\psi\ket{\phi_\beta}\bra{\phi_\beta}.\\
    \end{split}
\end{equation}

Now
\begin{equation}
    \Tr\left[\left(\frac{1}{2^n}\sum_{x\in V}W_x\psi W_x^\dagger\right)^2\right]=\Tr\left[\left(\sum_{\beta\in \F_2^n}\ket{\phi_\beta}\bra{\phi_\beta}\psi\ket{\phi_\beta}\bra{\phi_\beta}\right)^2\right]=\sum_\beta\bra{\phi_\beta}\psi\ket{\phi_\beta}^2,
\end{equation}
and
\begin{align*}
\Tr\left[\left(\frac{1}{2^n}\sum_{x\in V}W_x\psi W_x^\dagger\right)^2\right]
&=2^{-2n}\Tr\left[\sum_{x,y\in V}\psi W_x^\dagger W_y \psi (W_x^\dagger W_y)^\dagger\right] \\
&=2^{-n}\Tr\left[\sum_{y\in V}\psi W_y \psi  W_y^\dagger\right]
=\sum_{y\in V}p_{\psi}(x).
\end{align*}
\end{proof}

\section{Proof of Modified BSG Theorem}\label{app:BSG}

\BSG*

\begin{proof}
  The proof is based on \cite{JopLectureNote}.
  For $x\in S+S$, let $r(x)=\abs*{\lrs{(a,b)\in S^2: a+b=x}}$.
  Note that $\sum_{x\in S}r(x)= \abs{S}^2\PP_{s,s'\in S}\lrb{s+s'\in S} 
    \ge \vep\abs*{S}^2$.
  We uniformly randomly choose $Z$ from $S$.
  Let $B = S \cap (S+Z)$.
  Let $S=\lrs{(a,b)\in S^2: r(a+b)\le c\vep^2\abs*{S}}$.
  We have 
  \begin{equation}
    \Ex{\abs*{B}} =
      \f{1}{\abs*{S}} \sum_{z\in S} \abs*{S\cap (S+z)} 
    = \f{1}{\abs*{S}} \sum_{z\in S} r(z)
    \ge \vep \abs*{S},
  \end{equation}
  and 
  \begin{align}
    \Ex{\abs*{B^2\cap S}} &= 
      \sum_{(a,b)\in S} \Prob{a\in B\wedge b\in B}  \\
    &= \sum_{(a,b)\in S} \Prob{a\in S+Z\wedge b\in S+Z}  \\
    &\le \sum_{(a,b)\in S} \f{r(a+b)}{\abs*{S}} \\
    &\le \sum_{(a,b)\in S} c\vep^2 \\
    &\le c\vep^2\abs*{S}^2.
  \end{align}
  To conclude, 
  \begin{equation}
    \Ex{2c\abs*{B}^2 - \abs*{B^2\cap S}} \ge c\vep^2 \abs*{S}^2.
  \end{equation}
  Therefore there exists $z$ and $B$, such that $\abs*{B}\ge 
    \f{\vep}{\sqrt{2}}\abs*{S}$ and $\abs*{B^2\cap S}\le 2c\abs*{B}^2$.
  Fix $c=\f{1}{16}$, then we have
    $B$ such that $\abs{B}\ge \f{\vep}{\sqrt{2}}\abs{S}$.
  Besides, let $\Gamma=(B\times B, E)$ be the bipartite graph with 
    edge set 
  \begin{equation}
    E = \lrs{(a,b)\in B^2: r(a+b)\ge \f{\vep^2}{16}\abs{S}}.
  \end{equation}
  Then $\abs{E}\ge \f{7}{8}\abs{B}^2$.
  Let $S'\subseteq B$ be the set of vertices with degree at least
    $\f{3}{4}\abs{B}$. 
  Since $\abs{E} = \sum_{x\in B}\text{deg}(x)$, we have
    $\abs{S'}\ge \f{1}{2}\abs{B} \ge \f{\vep}{2\sqrt{2}}\abs{S}$.
  Furthermore, for every pair $a,b\in S'$, they share at least 
    $\f{1}{2}\abs{B}$ common neighbors. 
  We will show $S'$ has small doubling number as follows.
  For each $u\in S'+S'$, let $(a_u,b_u)\in S'\times S'$ be
    a pair such that $a_u + b_u = u$.
  Then, the numbmer of quadruples $(a, b, c, d)\in S^4$ satisfying
    $(a+b)+(c+d)\in S'-S'$ is at least
  \begin{align}
    \sum_{u\in S'+S'}\sum_{x\in B} r(a_u+x) r(b_u+x) 
      &\ge \abs*{S'+S'}\p{\f{1}{2}\abs{B}}\p{\f{\vep^2}{16}\abs*{S}}^2 \\
    &\ge \f{1}{2\sqrt{2}} \vep^5 \abs*{S'+S'}\abs*{S}^3.
  \end{align}
  Since this number is at most $\abs*{S}^4$, we conclude that
  \begin{equation}
    \abs*{S'+S'} \le 2\sqrt{2}\vep^{-5}\abs{S} 
      \le 8\vep^{-6} \abs*{S'}.
  \end{equation}
\end{proof}

\end{document}